\newcommand{\paperversion}{arxiv}
\providecommand{\paperversion}{final}
\RequirePackage{paperversions}
 
\PassOptionsToPackage{dvipsnames}{xcolor}

\documentclass[runningheads]{llncs}

\usepackage[utf8]{inputenc}

\usepackage[T1]{fontenc}
\usepackage{microtype}
\usepackage[english]{babel}

\usepackage[scaled]{inconsolata}

\usepackage{tikz}

\usepackage{amsmath}
\usepackage{amssymb}
\usepackage{mathtools}
\mathtoolsset{mathic}

\usepackage{mathpartir}
\usepackage{pl-syntax}

\newif\ifappendixproofs
\appendixproofstrue

\iffinalformat
 \ifreport
  \appendixproofstrue
 \else
  \appendixproofsfalse
 \fi
\fi

\ifappendixproofs
  \usepackage[bibliography=common]{apxproof}
\else
  \usepackage[bibliography=common, appendix=strip]{apxproof}
\fi

\usepackage[pages=16]{checkpagelimit}
\ifshowcomments
  \usepackage[checkNoComments]{aplcomments}
\else
  \usepackage[disabled]{aplcomments}
\fi
\newcommenter{todo}{1.0,0.2,0.2}

\usepackage{listings}

\usepackage{hyperref}
  
\usepackage{cleveref}
  \ifreport
    \crefformat{section}{Section~#2#1#3}
    \Crefformat{section}{Section~#2#1#3}
  \else
    \crefformat{section}{\S#2#1#3}
  \fi

\usepackage{overrides}

\usepackage{defs}
\usepackage{concrete-syntax}
\usepackage{functions}
\usepackage{information-flow}
\usepackage{judgment}
\usepackage{meta-variables}
\usepackage{theorem}
\usepackage[pdf]{graphviz}

\usepackage[numbers,square,sort&compress]{natbib}
\renewcommand{\citet}[2]{#1~\cite{#2}}
\renewcommand{\Citet}[2]{#1~\cite{#2}}

\ifdefined\acmYear
\setcopyright{rightsretained}
\acmPrice{}
\acmDOI{}
\acmYear{2023}
\copyrightyear{2023}
\acmSubmissionID{}
\acmISBN{}
\acmConference
  [ACM CCS 2023]
  {ACM Conference on Computer and Communications Security}
  {November 26--30, 2023}
  {Copenhagen, Denmark}
\acmBooktitle{ACM Conference on Computer and Communications Security}
\fi

\newcommand\institution[1]{#1}
\newcommand\department[1]{\relax}
\newcommand\city[1]{, #1}
\newcommand\province[1]{ #1}
\newcommand\postcode[1]{ #1}
\newcommand\country[1]{, #1}

\newcommand{\cornell}[1]{\institution{Cornell University}\department{Department of Computer Science}\city{Ithaca}\province{NY}\postcode{14850}\country{USA}\\
  \email{#1}}

\newcommand{\observe}[1]{\institution{Observe, Inc.}\city{San Mateo}\province{CA}\postcode{94402}\country{USA}\\
  \email{#1}}

\ifblinded
\author{Anonymous Authors}
\institute{}
\else
\author{Silei Ren\inst{1}\orcidID{0000-0002-4182-0211}\and
  Co\c{s}ku Acay\inst{2}\orcidID{0000-0002-0487-1167}\and
  \\
  Andrew C. Myers\inst{1}\orcidID{0000-0001-5819-7588}}

\institute{
  \cornell{sr2262@cornell.edu,andru@cs.cornell.edu}
  \and
  \observe{coskuacay@gmail.com}
}
\fi
 
\title{
  An Algebraic Approach to Asymmetric Delegation and Polymorphic Label Inference
\ifreport
  {\\ \LARGE (Technical Report)}
\fi
}
\titlerunning{Asymmetric Delegation and Polymorphic Label Inference}
 
\ifshowpagenumbers
\fi

\draftwarning

\begin{document}

\maketitle

\begin{abstract}
  Language-based information flow control (IFC) enables reasoning about
and enforcing security policies in decentralized applications.
While information flow properties are relatively extensional and compositional,
designing expressive systems that enforce such properties
remains challenging.
In particular, it can be difficult to use IFC labels to model
certain security assumptions, such as semi-honest agents.
Motivated by these modeling limitations, we study the algebraic semantics of
lattice-based IFC label models, and propose a semantic framework
that allows formalizing asymmetric delegation,
which is partial delegation of confidentiality or integrity.
Our framework supports downgrading of information
and ensures their safety through nonmalleable information
flow (NMIF).
To demonstrate the practicality of our framework,
we design and implement a novel algorithm that statically
checks NMIF and a label inference procedure
that efficiently supports bounded label polymorphism,
allowing users to write code generic with respect to labels.
 \end{abstract}

\section{Introduction}
\label{sec:intro}
\emph{Information Flow Control (IFC)}~\cite{sm-jsac,iflow-properties}
is a well-established approach for enforcing information security.
Using \emph{labels}, IFC systems specify
fine-grained policies on information flow that can be fully or partly
enforced through compile-time analysis.
These policies articulate the confidentiality and integrity goals of IFC systems,
which are \emph{security properties}~\cite{iflow-properties}:
hyperproperties~\cite{cs08} that constrain the set of system behaviors.

The most prominent security property for IFC systems is
\emph{noninterference}~\cite{GM82}, but it is
too restrictive in practice.
A major challenge for adopting language-based IFC is providing
developers with expressive yet intuitive ways to specify their intended
security policies.
To capture more nuanced security policies,
the expressiveness of IFC systems is enhanced by
\emph{downgrading mechanisms} such as
\emph{declassification} of confidential information
and \emph{endorsement} of untrusted information.
Misuse of these mechanisms is further mitigated by
enforcing \emph{nonmalleable information flow}~\cite{nmifc} (NMIF),
a security property controlling downgrading.

\emph{Delegation} is another common mechanism for specifying security.
Delegation allows one principal to grant (delegate) power to another,
expressing that the first principal trusts the second.
Delegation can compactly represent important aspects of the system's security
policy: when there is delegation between two principals,
ensuring the delegator's security also necessitates
enforcing the delegatee's security.
Delegation is commonly supported not only in information flow control
systems~\cite{ml-ifc-97,jif, flam,flac},
but also in a wide range of enforcement mechanisms, including access
control~\cite{denning-lattice, integrity, rbac-iflow}
(where delegation is often referred to as a \emph{principal} or \emph{role hierarchy}),
authorization logics~\cite{abadi06, ccd08,  flam, focal, FOLFL},
and capability systems~\cite{li2003delegation, matetic2018delegatee}.

The expressive power of delegation can be increased through
what we call \emph{asymmetric delegation}:
fine-grained delegation of either \emph{confidentiality} or
\emph{integrity}.
Intuitively, when a principal Alice delegates her confidentiality to another
principal Bob, she allows Bob to observe all information visible to her.
When Alice delegates her integrity to Bob, she trusts that all information
accepted by Bob has not been maliciously modified.
With asymmetric delegation, we can model security settings like
the semi-honest trust assumption in cryptographic applications
and the security setting of blockchains.
In the semi-honest setting, principals trust each other to follow the protocol
(trust each other with integrity),
but do not trust each other with their secrets (but not confidentiality).
In the blockchain setting, principals do not trust each other to follow protocols,
but all information is public:
they effectively trust each other with respect to confidentiality.

While asymmetric delegation increases the expressive power of IFC systems,
its precise role—particularly in the presence of downgrading—remains poorly understood.
We address this gap by presenting a general and expressive semantic framework
for IFC labels that formalizes both asymmetric delegation and its interaction with downgrading.
Although prior work~\cite{flac, shtz06} develops IFC systems that
support certain forms of asymmetric delegation, these systems lack
sound and complete NMIF enforcement.
Building on our framework, we develop algorithms for
verifying the associated semantic security properties.

Experience with language-based security highlights the importance of
IFC label inference to reduce the burden on programmers \cite{jif, viaduct-pldi21}.
In addition, allowing programmers to write code that is generic with respect
to labels enhances modularity and code reuse.
We support such generic programming through an efficient label inference
procedure that supports bounded label polymorphism.

To evaluate our approach,
we update the label model of the Viaduct compiler~\cite{viaduct-pldi21}
and extend its static information flow analysis.
Our implementation features a more concise and modular syntax for
specifying trust assumptions,
as well as a more efficient label inference procedure.

The rest of the paper is structured as follows:
\begin{itemize}
  \item \Cref{sec:overview} motivates asymmetric delegation using
  a semi-honest secure multiparty computation (MPC) program.

  \item \Cref{sec:principals} and \cref{sec:hyperproperty} study the effects
  of asymmetric delegation on security properties using a novel semantic framework.

  \item \Cref{sec:algorithm} presents algorithms that statically
  enforce the security properties.

  \item \Cref{sec:inference} introduces an inference procedure
  supporting bounded label polymorphism.

\end{itemize}

\begin{figure}
  \noindent
  \begin{minipage}[t]{.48\textwidth}
\begin{lstlisting}[style=custom-small]
host Alice : {A $\wedge$ B<-} |\label{line:viaduct-host-alice}|
host Bob   : {B $\wedge$ A<-} |\label{line:viaduct-host-bob}|

val a: {A $\wedge$ B<-} = Alice.input|\label{line:viaduct-input-alice}|
val b: {B $\wedge$ A<-} = Bob.input  |\label{line:viaduct-input-bob}|
val w: {A $\wedge$ B} = a > b |\label{line:viaduct-comparison}|

Alice.output(
  declassify w to {A $\wedge$ B<-}) |\label{line:viaduct-decl-alice}|
Bob.output(
  declassify w to {B $\wedge$ A<-}) |\label{line:viaduct-decl-bob}|
\end{lstlisting}
    \caption{Yao's Millionaires' problem in Viaduct~\cite{viaduct-pldi21}.
The programmer must manually assign labels to hosts.}
    \label{fig:semi-honest-mpc-viaduct}
  \end{minipage}
  \hfill
  \begin{minipage}[t]{.48\textwidth}
\begin{lstlisting}[style=custom-small]
host Alice, Bob |\label{line:yao-hosts}|
assume Alice = Bob for integrity|\label{line:yao-delegation}|

val a: {Alice} = Alice.input |\label{line:yao-input-alice}|
val b: {Bob}   = Bob.input |\label{line:yao-input-bob}|
val w: {Alice $\join$ Bob} = a > b |\label{line:yao-comparison}|

Alice.output(
  declassify w to {Alice}) |\label{line:yao-decl-alice}|
Bob.output(
  declassify w to {Bob})   |\label{line:yao-decl-bob}|
\end{lstlisting}
    \caption{Yao's Millionaires' problem implemented with delegation.
{\small $\alice \join \bob$} is shorthand for
      {\small $\langle \alice \wedge \bob, \alice \vee \bob \rangle$}.}
    \label{fig:semi-honest-mpc}
  \end{minipage}
\end{figure}

\section{A Case for Delegation}
\label{sec:overview}

\subsection{Semi-Honest Attackers in Cryptography}

Asymmetric delegation can capture a wide variety of security settings.
Already mentioned is the semi-honest threat model,
widely studied in the cryptography literature~\cite{yao82}.
In this model, principals correctly follow the protocol,
but attempt to improperly learn other principals' secrets.
Modeling the semi-honest setting in IFC systems remains a challenge.
We first give an example of modeling semi-honest security in
Viaduct~\cite{viaduct-pldi21},
a state-of-the-art compiler that translates information flow policies to
cryptographic protocols.
We then illustrate how delegation improves usability and modularity.

Consider Yao's well-known Millionaires' Problem~\cite{yao82},
where Alice and Bob wish to compare their wealth
without revealing actual numbers.
\Cref{fig:semi-honest-mpc-viaduct} shows a Viaduct implementation.
\Cref{line:viaduct-host-alice,line:viaduct-host-bob}
declare the hosts \lstinline{Alice} and \lstinline{Bob}
and assign them information flow labels that capture the security assumptions.
The hosts are assigned different and incomparable confidentiality labels
(\lstinline{A} for \lstinline{Alice} and \lstinline{B} for \lstinline{Bob})
but the same integrity label (\lstinline{A $\wedge$ B})
to reflect the trust relation in the semi-honest model.
\Cref{line:viaduct-input-alice,line:viaduct-input-bob} gather input from the hosts;
input from a host has the same label as that host.
\Cref{line:viaduct-comparison} stores the result of the comparison in
\lstinline{w}, which has a label following standard IFC rules:
the result of a computation is more secret and less trusted than all of its inputs.
Specifically, \lstinline{w} has a confidentiality of \lstinline{A $\wedge$ B}
since it is derived using secret data from both hosts,
and has an integrity of \lstinline{A $\wedge$ B} since that is the integrity
of both inputs.
Finally, \cref{line:viaduct-decl-alice,line:viaduct-decl-bob} output \lstinline{w}
to \lstinline{Alice} and \lstinline{Bob}, respectively.
Note that sending \lstinline{w} to \lstinline{Alice} leaks information about
\lstinline{Bob}'s secret data (\lstinline{b}), which violates noninterference.
Viaduct requires an explicit \lstinline{declassify} statement to indicate that
information leakage is intentional.

\subsection{Modeling Security with Delegation}
Viaduct models security by encoding trust into labels, but this
approach has problems.
First, programmers must encode security assumptions by carefully crafting host labels,
which becomes tricky in large systems with many assumptions.
Second, this encoding pollutes the entire program.
In \cref{fig:semi-honest-mpc-viaduct}, every label annotation must acknowledge
the semi-honest assumption by carrying around additional integrity
(i.e., \lstinline{A<-} or \lstinline{B<-}).
And third, the encoding breaks modularity.
For example, to add a new host \lstinline{Chuck} to the program,
we would need to edit every label annotation to carry an
extra integrity component of \lstinline{C<-},
requiring changes throughout
the program even though \lstinline{Chuck}
is not involved in this portion of the computation.
Delegation addresses all of these issues.

\Cref{fig:semi-honest-mpc} implements Yao's Millionaires' Problem using delegation.
Hosts are no longer assigned cryptic information flow labels;
instead, \cref{line:yao-delegation} directly states the security assumption:
\lstinline{Alice} and \lstinline{Bob} trust each other for integrity.
Variables have intuitive labels that do not need to repeat the semi-honest
security assumption:
input from \lstinline{Alice} has label \lstinline{Alice}.
Finally, adding a new host \lstinline{Chuck} requires no edits to existing code;
we only need to add the following lines:\footnote{
In fact, we could even support separate compilation as the program need not be
type-checked again:
a program considered secure with fewer assumptions is secure
with more assumptions.}\begin{lstlisting}[style=custom-small]
  host Chuck
  assume Alice = Chuck for integrity
\end{lstlisting}

\subsection{Nonmalleable Information Flow}

Downgrading statements (\lstinline{declassify} and \lstinline{endorse})
deliberately violate noninterference, so their unrestricted use poses a threat
to security.
Prior work~\cite{zm01b, msz06} identifies cases where the attacker can exploit
downgrading to gain undue influence over the execution,
and proposes \emph{robust declassification} and \emph{transparent endorsement}
to limit such cases.

Robust declassification requires that untrusted data is not declassified,
and transparent endorsement requires that secret data is not endorsed.
NMIF combines these two restrictions, which are key to
enabling the Viaduct compiler to securely instantiate programs with
cryptography~\cite{viaduct-pldi21}.

Here, ``secret'' and ``trusted'' are relative to a given attacker,
and NMIF must hold for all attackers.
In practice, the program cannot be type-checked separately for
every possible attacker, so a conservative condition is enforced:
downgraded data must be at least as trusted as it is secret.
For our example program, this condition means \lstinline{w} must have integrity
stronger than or equal to its confidentiality.
This condition is immediate in Viaduct since \lstinline{w} has label
\lstinline{$\langle$Alice $\wedge$ Bob, Alice $\wedge$ Bob$\rangle$}
in \cref{fig:semi-honest-mpc-viaduct}.
On the other hand, the same variable \lstinline{w} in \cref{fig:semi-honest-mpc}
has label \lstinline{$\langle$Alice $\wedge$ Bob, Alice $\vee$ Bob$\rangle$},
which seemingly has weaker integrity than confidentiality
(logically, \lstinline{Alice $\vee$ Bob} does \emph{not} imply
\lstinline{Alice $\wedge$ Bob}).
However, $\alice = \bob$ for integrity, so this label is equivalent to
\lstinline{$\langle$Alice $\wedge$ Bob, Alice $\wedge$ Bob$\rangle$}
using the following derivation:
  \[\alice \vee \bob
  =
  \alice \vee \alice
  =
  \alice = \alice \wedge \alice = \alice \wedge \bob \]
Delegation necessitates equational reasoning under assumptions,
and NMIF creates an interaction between confidentiality and integrity.
The combination of these two features is what makes asymmetric delegation tricky:
the cleaner syntax comes at the cost of additional technical complexity.
The following sections tame this complexity by developing a semantic
framework for labels, and algorithms that follow the semantics.

\section{Semantic Framework}
\label{sec:principals}

\subsection{The Lattice of Principals}

We build our semantic framework upon the \emph{lattice of principals},
used in prior work in authorization logics and information flow systems~\cite{ml-tosem, shtz06, jif, viaduct-pldi21, flam, dclabels, flac}.

A principal $\p \in \p!$ refers to an entity in decentralized systems
that can be either concrete, such as users or server machines,
or abstract, such as RBAC roles~\cite{rbac-iflow}
or quorums~\cite{zm14-plas}.
In IFC systems, they are often used as labels to annotate
policies on use of information~\cite{flam}.
For example, \lstinline{a: Alice} in \cref{fig:semi-honest-mpc}
requires the variable \lstinline{a} to only be written by
principals that can influence \lstinline{Alice}'s data,
and to remain secret to principals
who cannot observe \lstinline{Alice}'s information.

Principals are ordered by authority.
When \q delegates trust to \p,
we say \p \emph{acts for} \q,
written as $\p \actsfor \q$.
Conjunction (the ``and'' logic connective)
between principals $\p \wedge \q$ represents the
least combined authority of $\p$ and $\q$,
and disjunction (``or'')
$\p \vee \q$ represents greatest common authority.
The maximum authority $\strongest$ acts for all other authorities,
and the minimum authority $\weakest$ trusts all other authorities.
More authority is associated with elements lower in the lattice:
\(
  \strongest
  \actsfor
  \p \wedge \q
  \actsfor
  \p
  \actsfor
  \p \vee \q
  \actsfor
  \weakest
\).
\footnote{
It might seem odd to represent maximum authority as $\strongest$
and minimum authority with $\weakest$, since
some prior work (e.g., \cite{flam}) makes the opposite choice.
An intuitive justification:
the ``false'' proposition entails everything,
so no real principal can have authority $\strongest$.
All principals are trusted with $\weakest$.}

Additionally, we use a \emph{delegation context}, with the form
\( \ctx = \p_1 \actsfor \q_1, \cdots, \p_n \actsfor \q_n \),
to specify delegations that are not implied by the logical
structure of the principal lattice.
The declaration \lstinline{assume Alice = Bob for integrity} from
\cref{fig:semi-honest-mpc} is an example of a delegation context,
specifying both $\alice \actsfor \bob$ and $\bob \actsfor \alice$.
Using the delegation context is compatible with much prior work in IFC.
For example,
the \emph{trust configuration} from FLAM~\cite{flam},
\emph{meta-policies} from the Rx model~\cite{shtz06},
\emph{interpretation function} from DLM~\cite{myers-phdthesis},
and \emph{authority lattice} from label algebra~\cite{lblalgebra}
are all delegation contexts written differently.

The lattice of principals can be interpreted as an
authorization logic~\cite{abadi06, garg2006non}
where each principal is a proposition about
authorization policy.
As authorization logics are often
built upon propositional intuitionistic logics,
whose algebraic models are \emph{Heyting algebra}~\cite{Rutherford65},
we assume \p! is distributive.\footnote{
  A Heyting algebra is a distributive lattice that supports
  the \emph{relative pseudocomplement} ($\trustimply$) operation.
We do not need the $\trustimply$ operator until label inference.
}

\subsection{Delegation and Attackers}
In the MPC example, a delegation context makes some principals equivalent.
In this subsection, we give delegation a precise semantics.

In systems with decentralized trust,
all principals see other principals as potential attackers.
In the extreme case where no principal trusts another,
the attacker with respect to each principal controls all other principals.
Formally, we characterize an attacker $\sadv \subseteq \p!$
by the set of principals it controls.

To trust a principal is to disregard the case where it is the attacker.
Conversely, when a principal is attacker-controlled,
so are the principals it acts for. Formally:

\begin{definition}[Consistent Attackers]
  \label{def:consistent-attacker}
  $\sadv$ is consistent with $\ctx$ (\(\sadv \models \ctx\)) when:
  \[
      \forall \p, \q \in \p! \centerdot
      ((\p \actsfor \q) \vee (\p \actsfor \q) \in \ctx) \implies
      (\p \in \sadv \implies \q \in \sadv)
  \]
\end{definition}

The set of consistent attackers is an \emph{attacker model}:
\(\restrict{\sadv!}{\ctx} =
    \setdef
      {\sadv \in \sadv!}
      {\sadv \models \ctx}\).

The semantic trust levels of principals can be compared based on
the set of consistent attackers that control the principals. Formally:

\begin{definition}[Acts-for Semantics]
  \label{def:acts-for-delegation}
  \p acts for \q (written \( \ctx \models \p \sactsfor \q\)) when:
  \[
    \forall \sadv \in \restrict{\sadv!}{\ctx} \centerdot
    \p \in \sadv \implies \q \in \sadv
  \]
\end{definition}

We use ``$\sactsfor$'' to denote the semantics of the acts-for relation
``$\actsfor$''.
When viewing principals as authorization propositions,
``acts-for'' stands for ``implies'', and a delegation context is a theory
(list of propositions).
Each consistent attacker is a consistent interpretation (truth assignment)
of the principals and the delegation context,
where 1 is assigned to the principals the attacker controls.

In fact, a delegation context determines a \emph{congruence relation}\footnote{A congruence is an equivalence relation that preserves the lattice structure.}
over the lattice of principals,
where
\( \p \equiv_{\ctx} \q \)
is defined as
\(
  (\ctx \models \p \sactsfor \q)
  \wedge
  (\ctx \models \q \sactsfor \p)
\).
As a result, \( \equiv_{\ctx} \)
induces a quotient lattice \( \quotient{\p!}{\equiv_{\ctx}} \)
where mutually delegating principals are in the same equivalence class.
This quotient lattice is precisely the
\emph{Lindenbaum algebra}~\cite{Blok1989-BLOAL} of the theory \ctx:
the ``smallest'' algebraic model \p! where \ctx holds.

\begin{theoremrep}[Algebraic Model]
  \label{thm:correctness}
    The algebraic model of the principal lattice $\p!$ under delegation context
    $\ctx$ is the quotient lattice \( \quotient{\p!}{\equiv_{\ctx}} \).\ifreport\relax\else
      \footnote{Full proofs of all theorems from this paper are available in the technical report~\cite{ifc-delegation-tr}.}
    \fi
\end{theoremrep}
\begin{proof}
  Follows from the definition of Lindenbaum--Tarski Algebra
of propositional intuitionistic logic.
The Lindenbaum algebra of theory $\ctx$, $\quotient{\p!}{\equiv_{\ctx}}$,
is the initial algebra of the category of Heyting
algebras that are consistent with $\ctx$.
In human language, there is a lattice homomorphism from
$\quotient{\p!}{\equiv_{\ctx}}$ to every algebraic model $\mathbb{J}$ of $\ctx$.
 \end{proof}

The semantics of consistent attackers makes them the \emph{prime filters}
of the lattice of principals.
This is a result of the
Stone's Representation Theorem of Distributive Lattices \cite{Stone1938},
which says elements from distributive lattices can be fully characterized
by their prime filters.

\begin{theoremrep}[Attacker Model]
  \label{thm:attacker}
    $\restrict{\sadv!}{\ctx}$ is the set of
    \emph{prime filters} of \( \quotient{\p!}{\equiv_{\ctx}} \).
\end{theoremrep}
\begin{proof}
  Since each attacker is a truth assignment of some Heyting algebra $\mathbb{J}$,
it can be characterized by a homomorphism from $\mathbb{J}$ to the
two-point lattice $2 = \{\top, \bot\}$.
Therefore, each attacker is uniquely characterized by
a lattice homomorphism from $\quotient{\p!}{\equiv_{\ctx}}$ to 2,
which can in turn be characterized by its kernel (the set of $\mathbb{J}$
being mapped to $\top$).
By construction, each kernel is a prime filter
of $\quotient{\p!}{\equiv_{\ctx}}$.
 \end{proof}

Prime filters have intuitive interpretations,
which abstracts and generalizes
attacker models from prior work \cite{mpc-adversary-structure, nmifc}.
\( \sadv \subseteq \p! \) is a prime filter when:
  \begin{itemize}
    \item
      \( \weakest \in \sadv \):
      All attackers control the weakest authority;

    \item
      \( \strongest \not\in \sadv \):
      No attacker controls the strongest authority;

    \item
      If \( \p \actsfor \q \) and \( \p \in \sadv \), then \( \q \in \sadv \):
      Attackers are consistent;

    \item
      If \( \p \in \sadv \) and \( \q \in \sadv \),
      then \( \p \wedge \q \in \sadv \):
      An attacker controls the least combined authority of
      principals it controls;

    \item
      If \( \p \vee \q \in \sadv \),
      then either \( \p \in \sadv \) or \( \q \in \sadv \):
      If two principals are not controlled by an attacker,
      neither is their greatest common authority.\end{itemize}

\subsection{Labels}
\label{sec:labels}

Information flow systems mainly consider two aspects of security:
confidentiality (read authority) and integrity (write authority).
To express differing confidentiality and integrity policies,
we use \emph{the lattice of labels},
which are pairs of principals \( \lbl! = \p! \times \p! \).
For a label \( \lbl = \pair{\cc}{\ic} \),
\cc represents confidentiality and \ic represents integrity.
Like principals, each label reflects the authority required for a principal to
access information.
For example, information labeled $\pair{\strongest}{\weakest}$ can be read by
no principal except the strongest \strongest, but it can be influenced by any principal.

A label $\lbl$ act for $\lbl'$ when both $\lbl$ acts for $\lbl'$ for both
confidentiality and integrity.
Similarly, conjunction/disjunction of labels is defined by the
conjunction/disjunction of their confidentiality and integrity.
An asymmetric delegation context is a pair of different delegation
contexts $\Ctx = \pair{\cctx}{\ictx}$.

In general, \emph{asymmetric attackers} $\adv \in \adv! = \sadv! \times \sadv!$
may control different principals for confidentiality and integrity.
We write \( \adv = \pair{\cadv \in \sadv!}{\iadv \in \sadv!} \),
where \cadv represents the principals \adv controls for confidentiality,
and \iadv those for integrity.
Consequently, the attacker model $\adv!$ is a pair of prime filters.

As visualized in \cref{fig:ifc-lattice},
each attacker defines secret ($\secret$), public ($\public$),
trusted ($\trusted$) and untrusted ($\untrusted$) sets over
the lattice of labels.
\begin{align*}
\public_{\langle \cadv, \iadv \rangle} &=
\setdef{\langle \cc, \ic \rangle}{\cc \in \cadv}~,~~
\untrusted_{\langle \cadv, \iadv \rangle} =
\setdef{\langle \cc, \ic \rangle}{\ic \in \iadv}, \\
\secret_{\langle \cadv, \iadv \rangle} &=
\setdef{\langle \cc, \ic \rangle}{\cc \not\in \cadv}~,~~
\trusted_{\langle \cadv, \iadv \rangle} =
\setdef{\langle \cc, \ic \rangle}{\ic \not\in \iadv}
\end{align*}

\begin{figure}
    \centering
    \begin{tikzpicture}

\draw[thick] (0,2) -- (2,0) -- (0,-2) -- (-2,0) -- cycle;

\node at (0,2.3) {$\langle \strongest,\weakest \rangle$};
  \node at (2.5,0) {$\langle \weakest,\weakest \rangle$};
  \node at (-2.5,0) {$\langle \strongest,\strongest \rangle$};
  \node at (0,-2.3) {$\langle \weakest,\strongest \rangle$};

\draw (-0.7, 1.3) -- (1.3,-0.7);
  \draw (-1,-1) -- (1,1);
  \fill (0.3, 0.3) circle (1.5pt);

\draw[->] (-2, 0.4) -- (-0.4,2)  node[midway, rotate=45, above] {Integrity};
  \draw[->] (-2, -0.4) -- (-0.4,-2) node[midway, rotate=-45, below] {Confidentiality};

\draw[dotted] (-2, 0) -- (2, 0);

\node at (-0.85, 0.2) {$\mathcal{S\cap T}$};
  \node at (1.15, 0.2) {$\mathcal{P\cap U}$};
  \node at (0.15, -0.85) {$\mathcal{P\cap T}$};
  \node at (0.15, 1.15) {$\mathcal{S\cap U}$};

\draw[->] (-2, 2.6) -- (2, 2.6) node[midway, above] {Authority ($\sactsfor$)};
  \draw[->] (3, -2) -- (3, 2) node[midway, rotate=-90, above] {Information flow ($\flowsto$)};

\end{tikzpicture}
     \caption{The lattice of labels $(\lbl!, \sactsfor)$ and
    the lattice of information flow $(\lbl!, \flowsto)$ share the same
    underlying set, but use a different ordering.
The dotted line depicts labels with equal confidentiality and integrity:
    strictly above are compromised labels,
    and on or below are uncompromised labels.
}
    \label{fig:ifc-lattice}
\end{figure}

\section{Security Hyperproperties and Delegation}
\label{sec:hyperproperty}

Our semantic framework formalizes a key insight relating the attacker model and
delegation: more delegation means fewer attackers.
In turn, fewer attackers should make it easier for programs to be considered secure.
To express delegation, we extend prior definitions of IFC security
\emph{hyperproperties}~\cite{cs08} with our semantic framework.
Rather than a standard proof of security for an IFC-typed core calculus,
we abstract away from computation to concentrate on the core judgment in IFC systems:
when it is safe to relabel information, under a delegation context.

The simple system has states $\state \in \state!$, intentionally left unspecified.
An execution of the system emits a trace \trace, which is a sequence of states.
Define the \emph{behavior} $\behav \in \mathbb{B}$
of a program to be the \emph{set} of possible
execution traces it can emit.
A hyperproperty $\HP \subseteq \mathbb{B}$ is a set of behaviors.
A program satisfies a hyperproperty when its behavior is a member of
the hyperproperty.

In decentralized IFC systems,
hyperproperties $\HP_{\adv}$ are parameterized by the choice of attacker:
a program secure against one attacker may be insecure against another.
Let the hyperproperty $\HP_{\adv!}$ be the hyperproperty that characterizes
programs that are secure against all possible attackers from $\adv!$.
A behavior $B$ of a program falls into $\HP_{\adv!}$ precisely when
$B \in \HP_{\adv}$ for all $\adv \in \adv!$:
\begin{mathpar}
  \HP_{\adv!} = \bigcap_{\adv \in \adv!} \HP_{\adv}

  \Ctx(\HP_{\adv!}) = \bigcap_{\adv \in \restrict{\adv!}{\Ctx}} \HP_{\adv}
\end{mathpar}

When a program assumes a static delegation context $\Ctx$,
the attacker model is further restricted to the ones
consistent with $\Ctx$.
Therefore, a delegation context can be understood as
a hyperproperty transformer.
It follows that hyperproperties accept at least as many programs
after a delegation context is added.
\begin{theorem}
$\mathbb{HP}_{\adv!} \subseteq \Ctx(\mathbb{HP}_{\adv!})$.
\end{theorem}

This matches our intuition about static delegation:
the more trust assumptions, the fewer reasonable attackers,
the more programs considered secure.
\subsection{Noninterference and the Lattice of Information Flow}

For confidentiality, noninterference~\cite{sm-jsac, denning-lattice} demands that
information should not flow from high
to low (secret to public).
Noninterference of integrity requires that information should not
flow from low to high (untrusted to trusted).
To illustrate how delegation affects noninterference,
we adapt \citet{Clarkson and Schneider's}{cs08} definition of
\emph{observational determinism}~\cite{GM82}.

\begin{definition}[Observational Determinism]
  \label{def:OD}
  Let $\low$ be any set of \emph{low} labels.
  Observational determinism $\OD$ is the hyperproperty:
  
\begin{align*}
\OD_{\low} = \{B \mid \forall t_1, t_2 \in B \centerdot
t_1^0 =_{\low} t_2^0 \implies t_1 \approx_{\low} t_2 \}
\end{align*}
\end{definition}

State $t^0$ denotes the initial state of the trace $t$.
We leave the definition of low equivalence between states unspecified,
as in prior work on knowledge-based security
\cite{landauer93, 6266168, li2021generalpurposedynamicinformationflow, soloviev2023security}.

Noninterference for confidentiality
$\CNI_{\adv} = \OD_{\public_{\adv}}$
and integrity
$\INI_{\adv} = \OD_{\trusted_{\adv}}$
are mere instantiations of observational determinism
over public and trusted labels for some attacker $\adv$.
For an attacker model $\restrict{\adv!}{\Ctx}$:
\[ \Ctx(\NI_{\adv!}) =\bigcap_{\adv \in \restrict{\adv!}{\Ctx}}
{(\OD_{\public_{\adv}} \cap
\OD_{\trusted_{\adv}})} \]

Prior work \cite{GM84, iflow-properties, ml-ifc-97}
enforces low equivalence by ensuring that low-labeled information
is not influenced by high-labeled information.
Concretely, a dynamic \emph{relabel} is safe when
it does not relabel high-labeled information to a low label.
We formalize safe relabeling
as the \emph{flows-to} relation.

\begin{definition}[Flows-to]
  \label{def:flows-to}
  Label $\lbl$ securely flows to $\lbl'$ ($\Ctx \models \lbl \flowsto \lbl'$) when:
  \[
    \forall \adv \in \restrict{\adv!}{\Ctx} \centerdot
    (\lbl' \in \public_{\adv} \implies \lbl \in \public_{\adv})
      \wedge
    (\lbl' \in \trusted_{\adv} \implies \lbl \in \trusted_{\adv})
  \]
  Equivalently,
  \(
    \pair{\cctx}{\ictx}
    \models
    \pair{\cc}{\ic} \flowsto \pair{\cc'}{\ic'}
  \)
  when
  \( \cctx \models \cc' \sactsfor \cc \)
  and
  \( \ictx \models \ic \sactsfor \ic' \).
\end{definition}

As visualized in~\cref{fig:ifc-lattice},
flows-to and the acts-for define two lattices
on the same underlying set of labels.
The Lattice of Information operators are given by:
  \[
  \pair{\cc_1}{\ic_1} \sqcup \pair{\cc_2}{\ic_2} =
  \pair{\cc_1 \wedge \cc_2}{\ic_1 \vee \ic_2} ~~~~~~~
  \pair{\cc_1}{\ic_1} \sqcap \pair{\cc_2}{\ic_2} =
  \pair{\cc_1 \vee \cc_2}{\ic_1 \wedge \ic_2}
\]

\subsection{Downgrading and Nonmalleable Information Flow}
\label{sec:dynamic-policy}

Some programs, such as the MPC example from \cref{fig:semi-honest-mpc},
intentionally break noninterference.
Prior work on dynamic security policies either
achieve downgrading by \emph{relabeling} or by \emph{dynamic delegation}.
Visually, relabeling moves information downward in
\cref{fig:ifc-lattice} and dynamic delegation
moves the attacker partition leftward.

\subsubsection{Nonmalleable Information Flow (NMIF)}
\label{sec:nmifc}

To prevent misuse of downgrades by relabeling,
\citet{Cecchetti et al.}{nmifc} propose NMIF,
a security hyperproperty that combines robust declassification~\cite{zm01b}
with transparent endorsement.

Robust declassification requires that secret information
flow to public only when the information is trusted.
This restriction ensures that attackers do not influence disclosure of
information to them.
A declassification is only robust
when it declassifies secret--trusted information ($\secret_{\adv} \cup \trusted_{\adv}$).

Transparent endorsement is a dual condition that allows untrusted
information to influence trusted information only when
the information is public to the attacker.
It only allows endorsement of public--untrusted information
($\public_{\adv} \cap \untrusted_{\adv}$).

Therefore, secret--untrusted information should not be downgraded.
Indeed, the key recipe to enforcing NMIF is to enforce
noninterference for public or trusted information~\cite{nmifc}:
\[\NI^{\uncompromised{}}_{\adv!} = \OD_{\trusted_{\adv} \cup \public_{\adv}}\]

A label is \emph{compromised}
when it is secret--untrusted for some attacker~\cite{nmifc, zsm19}.
To enforce NMIF, it suffices to reject downgrading information
with compromised labels, so that there is no flow out from compromised labels.

Unfortunately, NMIF against $\restrict{\adv!}{\ctx}$ rejects all downgrades.
Namely, all labels (except for $\pair{\weakest}{\strongest}$) are compromised
against the attacker \( \adv_{\top} = \pair{\{\weakest\}}{\p!} \)
(it controls every principal for integrity and controls no principal for confidentiality).
Therefore, further restrictions on the attacker model are needed.

\subsubsection{NMIF Attackers}
\label{sec:nmifc-attacker}

Prior work \cite{zm01b, zsm19} assumes
attackers control more confidentiality than integrity.
We call them \emph{valid attackers}:

\begin{definition}[Valid Attackers]
  \label{def:valid-attackers}
  \(
    \validattackers =
    \setdef
      {\pair{\cadv}{\iadv} \in \adv!}
      {\iadv \subseteq \cadv}
  \)
\end{definition}

Restricting attackers to valid ones makes our framework mirror attacker models
studied in the cryptography literature~\cite{yao82},
where a principal is honest (not controlled by attacker),
semi-honest (controlled by attacker for confidentiality),
or malicious (controlled by attacker for both confidentiality and integrity).
The valid attacker restriction excludes
unrealistic ``malicious but incurious'' attackers.

In fact, much existing work satisfies the valid-attacker assumption by construction.
For example, robust declassification \cite{zm01b} originally
defines integrity and confidentiality by equivalence relations over system states.
The state transitions an active attacker may perform are, by construction,
observable by the attacker.

\begin{definition}[Uncompromised Labels]
  \label{def:uncompromised-labels}
  Label \lbl is uncompromised under \Ctx,
  written \( \Ctx \models \uncompromised{\lbl} \),
  when \lbl is either public or trusted for all \emph{valid attackers}:
  \[
    \forall \adv \in \restrict{\validattackers}{\Ctx} \centerdot
    \lbl \in \public_{\adv} \cup \trusted_{\adv}
  \]
\end{definition}

It follows that labels of principals are uncompromised:
\(\Ctx \models \uncompromised{\pair{\p}{\p}}\).

Uncompromised labels have an alternative characterization:
they are the labels with at least as much integrity as confidentiality.
Of course, in the presence of asymmetric delegation,
we cannot directly compare a label's confidentiality and integrity components
since each component has a different set of delegations.
We circumvent this problem by introducing a witnessing principal $\rr$
who has no more integrity than $\ic$ and no less confidentiality than $\cc$.

\begin{toappendix}
  \begin{definition}[Closures]
  Define upward and downward closed sets as follows:
  \begin{mathpar}
    \upset{\p} = \setdef{\q \in \p!}{\p \actsfor \q}

    \downset{\p} = \setdef{\q \in \p!}{\q \actsfor \p}
\end{mathpar}
\end{definition}

\begin{definition}[Equivalence Classes]
  The equivalence class of a principal with respect to a delegation context is:
  \[
    \closure[\ctx]{\p}
    =
    \setdef{q \in \p!}{
      (\ctx \models \p \sactsfor \q)
      \wedge
      (\ctx \models \q \sactsfor \p)
    }
  \]
We extend this notation to denote the equivalence closure of a set:
  \[
    \closure[\ctx]{\p*}
    =
    \bigcup_{\p \in \p*}{\closure[\ctx]{\p}}
  \]
\end{definition}

\begin{remark}
  We denote the \emph{set} of equivalence classes using standard notation:
  \[
    \p* / \ctx
    =
    \p* / {\equiv_{\ctx}}
    =
    \setdef{\closure[\ctx]{\p}}{\p \in \p*}
  \]
  for \( \p* \subseteq \p! \) (we might have \( \p* = \p! \)).
This is different from \( \closure[\ctx]{\p*} \subseteq \p! \).
\end{remark}

\begin{remark}
  Note that
  \( \closure[\ctx]{\p} \in \p! / \ctx \)
  and
  \( \upset{\closure[\ctx]{\p}} \subseteq \p! / \ctx \).
In particular,
  \( \upset{\closure[\ctx]{\p}} \)
  is a set of sets of principals, not a set of principals.
\end{remark}

\begin{theorem}[Prime Ideal/Filter]
  \label{thm:prime-filter}
  Let \( L \) be a distributive lattice,
  \( I \subseteq L \) an ideal,
  and \( F \subseteq L \) a filter
  such that
  \( I \cap F = \emptyset \).
Then there exists a prime ideal \( P \) and a prime filter \( Q \) such that
  \begin{mathpar}
  I \subseteq P

  F \subseteq Q

  P \cap Q = \emptyset
  \end{mathpar}
\end{theorem}

 \end{toappendix}

\begin{theoremrep}
  \label{thm:alternative-uncompromised}
  \(
    \pair{\cctx}{\ictx} \models \uncompromised{\pair{\cc}{\ic}}
    \iff
    \exists \rr \in \p! \centerdot
      (\ictx \models \ic \sactsfor \rr)
      \wedge
      (\cctx \models \rr \sactsfor \cc)
  \).
\end{theoremrep}
\begin{proof}
  We prove each direction separately.
\begin{pcases}
  \pcase{$\implies$}
    Assume, for contradiction, that
    \( \ic \not\in \closure[\ictx]{\closure[\cctx]{\downset{\cc}}} \).
The sets
    \( \upset{\closure[\ictx]{\ic}} \)
    and
    \( \closure[\cctx]{\downset{\cc}} / \ictx \)
    are disjoint
    and form a filter/ideal pair of \(\p! / \ictx \).
Thus, \cref{thm:prime-filter} gives a prime filter
    \( \iadv_0 \) of \( \p! / \ictx \) with
    \( \upset{\closure{\ic}} \subseteq \iadv_0 \)
    and
    \( \iadv_0 \cap (\closure[\cctx]{\downset{\cc}} / \ictx) = \emptyset \).

    Let
    \( \iadv = \bigcup{\iadv_0} \).
Note that
    \( \iadv \)
    and
    \( \closure[\ictx]{\closure[\cctx]{\downset{\cc}}} \)
    are disjoint and
    \(
      \closure[\cctx]{\downset{\cc}}
      \subseteq
      \closure[\ictx]{\closure[\cctx]{\downset{\cc}}}
    \),
    so
    \( \iadv \)
    and
    \( \closure[\cctx]{\downset{\cc}} \)
    are disjoint as well.
This in turn implies
    \( \iadv / \cctx \)
    and
    \( \downset{\cc} / \cctx \)
    are disjoint.
Moreover, these sets form a filter/ideal pair in \( \p! / \cctx \),
    so \cref{thm:prime-filter} gives a prime filter
    \( \cadv_0 \) of \( \p! / \cctx \)
    with
    \( \iadv / \cctx \subseteq \cadv_0 \)
    and
    \( \cadv_0 \cap (\downset{\cc} / \cctx) = \emptyset \).

    Now, define
    \( \cadv = \bigcup{\cadv_0} \)
    and note that
    \[
      \iadv
      =
      \bigcup{(\iadv / \cctx)}
      \subseteq
      \bigcup{\cadv_0}
      =
      \cadv
    \]
Since \( \iadv_0 \) is an ideal of \( \p! / \ictx \),
    we have \( \iadv \models \ictx \).
Since \( \cadv_0 \) is an ideal of \( \p! / \cctx \),
    we have \( \cadv \models \cctx \).
Combining these three results, we have
    \( \pair{\cadv}{\iadv} \in \restrict{\validattackers}{\pair{\cctx}{\ictx}} \).

    Because
    \( \pair{\cctx}{\ictx} \models \uncompromised{\pair{\cc}{\ic}} \)
    and
    \(
      \ic
      \in
      \upset{\closure[\ictx]{\ic}}
      \subseteq
      \iadv_0
\implies
\ic
      \in
      \iadv
    \),
    we have \( \cc \in \cadv \).
However, \( \cc \in \closure[\cctx]{\downset{\cc}} \),
    which contradicts the fact that
    \( \cadv_0 \) and \( \downset{\cc} / \cctx \) are disjoint.
Thus, we must have
    \( \ic \in \closure[\ictx]{\closure[\cctx]{\downset{\cc}}} \).

    Finally,
    \( \ic \in \closure[\ictx]{\closure[\cctx]{\downset{\cc}}} \)
    means there exists \( \rr \in \closure[\cctx]{\downset{\cc}} \)
    such that \( \ic \in \closure[\ictx]{\rr} \).
By definition, we have
    \( \ictx \models \ic \sactsfor \rr \).
Moreover,
    \(  \rr \in \closure[\cctx]{\downset{\cc}} \)
    means there exists \( \cc' \in \downset{\cc} \)
    such that \( \rr \in \closure[\cctx]{\cc'} \).
By definition,
    \( \cctx \models \rr \sactsfor \cc' \),
    and \cref{thm:attacker} with \( \cc' \in \downset{\cc} \)
    gives \( \cctx \models \cc' \sactsfor \cc \).
By transitivity, we have
    \( \cctx \models \rr \sactsfor \cc \).

  \pcase{$\impliedby$}
    Assume there exists \rr such that
    \( \ictx \models \ic \sactsfor \rr \)
    and
    \(  \cctx \models \rr \sactsfor \cc \).
We claim $\pair{\cc}{\ic}$ is public whenever it is untrusted.
Formally, let
    \( \pair{\cadv}{\iadv} \in \validattackers \)
    such that
    \( \pair{\cadv}{\iadv} \models \pair{\cctx}{\ictx} \),
    and assume \( \ic \in \iadv \).
We need to show \( \cc \in \cadv \).

    Since
    \( \iadv \models \ictx \),
    \( \ictx \models \ic \sactsfor \rr \),
    and
    \( \ic \in \iadv \),
    we have
    \( \rr \in \iadv \).
Moreover,
    \( \iadv \subseteq \cadv \) (\cref{def:valid-attackers}),
    so
    \( \rr \in \cadv \).
Finally, since
    \( \cadv \models \cctx \),
    \( \cctx \models \rr \sactsfor \cc \),
    and
    \( \rr \in \cadv \),
    we have
    \( \cc \in \cadv \)
    as desired.
\end{pcases}
 \end{proof}

Proofs can be found in the technical report~\cite{ifc-delegation-tr}.
In the absence of asymmetric delegation, \cref{thm:alternative-uncompromised}
reduces to a simple acts-for check ($\ctx \models \q \sactsfor \p$),
which prior systems rely on~\cite{zm01b, nmifc}.

\section{Algorithms}
\label{sec:algorithm}

In IFC systems that incorporate delegation, the acts-for relation $\ctx
\models \p_1 \sactsfor \p_2$ is frequently checked by label-inference

procedures~\cite{jif} and even at run time,
where it can impose significant overhead~\cite{sif07}.
Unfortunately, its semantic definition quantifies over the
potentially infinite set of all attackers,
which makes direct use of the definition infeasible in practice.
Similarly, the definition of uncompromised labels
$\Ctx \models \uncompromised{\lbl}$ does not yield an algorithm.

In this section, we assume oracle access to syntactic lattice operations
($\actsfor$, \strongest, \weakest, $\wedge$, $\vee$) of \p!,
and propose sound and complete algorithms
that check for the aforementioned relations.
Proofs are available in the technical report~\cite{ifc-delegation-tr}.

\subsection{Acts-for Algorithm}
\Cref{fig:actsfor-algorithm} gives an algorithm for deciding
$\ctx \models \p \sactsfor \q$.
We write $\ctx \vdash \p \sactsfor \q$ to denote this algorithmic system.

\begin{algorithm}[Acts-for $\ctx \vdash \p \sactsfor \q$]
  \label{fig:actsfor-algorithm}
  \begin{mathpar}
  \inferrule
    [\named{$\p!$-Axiom}{rule:p-actsfor-axiom}]
    {
      \p \actsfor \q
    }
    {
      \emptyList \vdash \p \sactsfor \q
    }

  \inferrule
    [\named{$\p!$-Delegation}{rule:p-actsfor-delegation}]
    {
      \ctx \vdash \p \wedge \q' \sactsfor \q
      \\
      \ctx \vdash \p \sactsfor \q \vee \p'
    }
    {
      \ctx, \p' \delegatedBy \q' \vdash \p \sactsfor \q
    }
\end{mathpar}
 The algorithm recursively applies the lattice axioms and
\cref{rule:p-actsfor-delegation} until the delegation context is
empty for all sub-cases.
The acts-for relation holds when \cref{rule:p-actsfor-axiom}
applies to all sub-cases.
\end{algorithm}

\begin{toappendix}
  \begin{lemma}
  \label{thm:context-extension}
  If
  \( \ctx_1 \models \p \sactsfor \q \),
  then
  \( (\ctx_1, \ctx_2) \models \p \sactsfor \q \)
  for any \( \ctx_2 \).
\end{lemma}
\begin{proof}
  Immediate since
  \(
    \restrict{\sadv!}{\ctx_1, \ctx_2}
    \subseteq
    \restrict{\sadv!}{\ctx_1}
  \).
\end{proof}

\begin{corollary}
  \label{thm:syntactic-to-semantic-acts-for}
  If
  \( \p \actsfor \q \),
  then
  \( \ctx \models \p \sactsfor \q \)
  for any \ctx.
\end{corollary}
\begin{proof}
  We have \( \emptyList \models \p \sactsfor \q \) by \cref{thm:attacker}.
  The result then follows by \cref{thm:context-extension}.
\end{proof}

\begin{lemma}
  \label{thm:no-delegation}
  If
  \( \emptyList \models \p \sactsfor \q \),
  then
  \( \p \actsfor \q \).
\end{lemma}
\begin{proof}
  Assume, for contradiction, that \( \p \nactsfor \q \).
The set
  \( \upset{\p} = \setdef{\p' \in \p!}{\p \actsfor \p'} \)
  is a filter of \p!, and
  \( \downset{\q} = \setdef{\q' \in \p!}{\q' \actsfor \q} \)
  is an ideal of \p!.
Moreover, \( \upset{\p} \) and \( \downset{\q} \) are disjoint since
  \( \p \nactsfor \q \).
By \cref{thm:prime-filter}, there exists a prime filter
  \( \p* \supseteq \upset{\p} \) that is also disjoint from \( \downset{\q} \).
This means \( \p \in \p* \) and \( \q \not\in \p* \).
However, \( \p* \in \adv! \) by \cref{thm:attacker}
  (attackers are prime filters)
  and \( \p* \models \emptyList \) trivially,
  which contradicts the assumption
  \( \emptyList \models \p \sactsfor \q \).
\end{proof}

\begin{lemma}
  \label{thm:delegation-to-meet}
  If
  \( (\ctx, \p_1 \delegatedBy \q_1) \models \p_2 \sactsfor \q_2 \),
  then
  \( \ctx \models \p_2 \wedge q_1 \sactsfor \q_2 \).
\end{lemma}
\begin{proof}
  Let $\sadv$ be an attacker such that $\sadv \models \ctx$ and
  $\p_2 \wedge q_1 \in \sadv$.
We need to show $\q_2 \in \sadv$.
Since $\p_2 \wedge q_1 \in \sadv$ and $\sadv$ is upward-closed
  (\cref{thm:attacker}), we have $\p_2, \q_1 \in \sadv$.
From $\sadv \models \ctx$ and $\q_1 \in \sadv$,
  we get $\sadv \models \ctx, \p_1 \delegatedBy \q_1$.
Using this, and the fact that $\p_2 \in \sadv$,
  we invoke our primary assumption to get $\q_2 \in \sadv$.
\end{proof}

\begin{lemma}
  \label{thm:delegation-to-join}
  If
  \( (\ctx, \p_1 \delegatedBy \q_1) \models \p_2 \sactsfor \q_2 \),
  then
  \( \ctx \models \p_2 \sactsfor \q_2 \vee \p_1 \).
\end{lemma}
\begin{proof}
  Let $\sadv$ be an attacker such that $\sadv \models \ctx$ and
  $\p_2 \in \sadv$.
We need to show $\q_2 \vee \p_1 \in \sadv$.
If $\p_1 \in \sadv$, then $\q_2 \vee \p_1 \in \sadv$ since $\sadv$
  is upward-closed (\cref{thm:attacker}), so assume $\p_1 \not\in \sadv$.
Then, $\sadv \models \ctx, \p_1 \delegatedBy \q_1$ and we assumed
  $\p_2 \in \sadv$, so we can invoke our primary assumption to get
  $\q_2 \in \sadv$.
Since $\sadv$ is upward-closed, this gives $\q_2 \vee \p_1 \in \sadv$.
\end{proof}
 \end{toappendix}

\begin{theoremrep}[Correctness of Acts-for]
  \label{thm:acts-for-correctness}
  \Cref{fig:actsfor-algorithm} terminates,
  and it is sound and complete with respect to
  \cref{def:acts-for-delegation}:
  \(
    \ctx \vdash \p \sactsfor \q
    \iff
    \ctx \models \p \sactsfor \q
  \).
\end{theoremrep}
\begin{proof}
  We first prove termination.
The algorithm terminates because each time \cref{rule:p-actsfor-delegation} is
applied, the delegation context gets smaller.
Since all delegation contexts are finite, \cref{rule:p-actsfor-delegation} is
applied at most finitely many times until it reaches the base case
\cref{rule:p-actsfor-axiom}.

We prove soundness by induction on the derivation of
\( \ctx \vdash \p \sactsfor \q \).

\begin{pcases}
  \prule{rule:p-actsfor-axiom}.
    We have \( \p \actsfor \q \) by inversion on the derivation.
The results follows from \cref{thm:syntactic-to-semantic-acts-for}.

  \prule{rule:p-actsfor-delegation}.
    We have
    \( \ctx = (\ctx', \p' \delegatedBy \q') \),
    \( \ctx' \vdash \p \wedge \q' \sactsfor \q \),
    and
    \( \ctx' \vdash \p \sactsfor \q \vee \p' \).
Induction gives
    \( \ctx' \models \p \wedge \q' \sactsfor \q \)
    and
    \( \ctx' \models \p \sactsfor \q \vee \p' \).
Let $\sadv$ be an attacker such that
    $\sadv \models \ctx', \p' \delegatedBy \q'$
    and
    $\p \in \sadv$.
We need to show $\q \in \sadv$.

    If $\p' \in \sadv$, then $\q' \in \sadv$ since
    $\sadv \models \p' \delegatedBy \q'$.
Since $\p, \q' \in \sadv$,
    \cref{thm:attacker} gives $\p \wedge \q' \in \sadv$.
The first induction hypothesis then gives $\q \in \sadv$.

    Otherwise, $\p' \not\in \sadv$.
The second induction hypothesis gives $\q \vee \p' \in \sadv$.
\Cref{thm:attacker} then implies $\q \in \sadv$
    since we cannot have $\p' \not\in \sadv$ and $\q \not\in \sadv$
    but $\q \vee \p' \in \sadv$.
\end{pcases}

We prove completeness by induction on $\ctx$.
\begin{pcases}
  \pcase{$\ctx = \emptyList$}
    Assume $\emptyList \models \p \sactsfor \q$.
Then $\p \actsfor \q$ by \cref{thm:no-delegation}.
Applying \cref{rule:p-actsfor-axiom}, we get
    $\emptyList \vdash \p \sactsfor \q$
    as desired.

  \pcase{$\ctx = \ctx', \p' \delegatedBy \q'$}
    Assume $\ctx', \p' \delegatedBy \q' \models \p \sactsfor \q$.
By \cref{thm:delegation-to-meet,thm:delegation-to-join}, we get
    $\ctx' \models \p \wedge \q' \sactsfor \q$
    and
    $\ctx' \models \p \sactsfor \q \vee p'$.
By induction, we get
    $\ctx' \vdash \p \wedge \q' \sactsfor \q$
    and
    $\ctx' \vdash \p \sactsfor \q \vee p'$.
We can now apply \cref{rule:p-actsfor-delegation} to get
    $\ctx', \p' \delegatedBy \q' \vdash \p \sactsfor \q$
    as desired.
  \qedhere
\end{pcases}
 \end{proof}

\subsection{NMIF Algorithm}

Neither the semantic definition of uncompromised labels
(\cref{def:uncompromised-labels}) nor their alternative characterization
(\cref{thm:alternative-uncompromised}) lends itself to an algorithmic implementation.
The semantic definition quantifies over all valid attackers (a potentially infinite set),
and the alternative characterization conjures up an intermediate principal.
Our solution is to use the alternative characterization
but with a ``best principal.''
For that, we use $\minrep{\p}$, the highest-authority principal
in the equivalence class of \p.\footnote{Recall that higher authority is lower in the authority lattice,
thus the use of $\min$ as opposed to $\max$.}

\begin{definition}[Minimal Principal]
  \label{def:min-element}
  \( \minrep[\ctx]{\p} \in \p! \)
  is the (necessarily) unique principal such that
  \( \ctx \models \p \sactsfor \q \)
  if and only if
  \( \minrep[\ctx]{\p} \actsfor \q \)
  for any $\q \in \q!$.
\end{definition}

We must demand more structure on \p! to compute $\minrep{\p}$.
Specifically, we rely on an oracle that returns \emph{join-prime} factorizations
of arbitrary principals.
Intuitively, a join-prime principal cannot be written as the join of other principals.
In finite lattices, these are the principals of the form
\( \p = \q_1 \wedge \dots \wedge \q_n \).
Factorization oracles arise trivially in existing implementations of IFC models,
as these are based on lattices with a finite set of principal names~\cite{jif, dclabels, flam, viaduct-pldi21}.

\begin{algorithm}[Min $\minrep{\p} = \q$]
  \label{fig:min-max-algorithm}
  \begin{mathpar}
  \inferrule
    [\named{Min-Base}{rule:min-base}]
    {
      \joinprime{\p}
      \\
      \forall (\p' \actsfor \q') \in \ctx \centerdot
        \p \not\actsfor \p'
    }
    {
      \minrep[\ctx]{\p} = \p
    }

  \inferrule
    [\named{Min-Pick}{rule:min-pick}]
    {
      \joinprime{\p}
      \\
      \p \actsfor \p'
    }
    {
      \minrep[\ctx, \p' \delegatedBy \q']{\p} = \minrep[\ctx]{\p \wedge \q'}
    }

  \inferrule
    [\named{Min-Factor}{rule:min-factor}]
    {
      \neg\joinprime{\p}
      \\
      \p = \p_1 \vee \dots \vee \p_n
      \\
      \forall i \in \enumSet{n} \centerdot \joinprime{\p_i}
    }
    {
      \minrep[\ctx]{\p} = \bigvee_{i \in \enumSet{n}}{\minrep[\ctx]{\p_i}}
    }
\end{mathpar}
 \end{algorithm}

The rules from \Cref{fig:min-max-algorithm} can be applied in any
order without backtracking because of the uniqueness of
$\minrep[\ctx]{\p}$.
Moreover, all derivations are finite since either the size of $\ctx$ decreases,
or we switch from a reducible element to a finite set of irreducible elements.

\begin{toappendix}
  \begin{lemma}
  \label{thm:meet-to-delegation}
  If
  \( \ctx \models \p_2 \wedge q_1 \sactsfor \q_2 \)
  and
  \( \ctx \models \p_2 \sactsfor p_1 \),
  then
  \( (\ctx, \p_1 \delegatedBy \q_1) \models \p_2 \sactsfor \q_2 \).
\end{lemma}
\begin{proof}
  Follows trivially from \cref{thm:attacker} since attackers are
  closed under $\wedge$.
\end{proof}
 \end{toappendix}

\begin{theoremrep}[Correctness of Min]
  \label{thm:min-max-correctness}
  \Cref{fig:min-max-algorithm} terminates,
  and it is sound and complete with respect to \cref{def:min-element}.
\end{theoremrep}
\begin{proof}
  We first prove termination.
Define the following relation $\succeq$ between expressions
in the form of $\minrep[\ctx]{\p}$:
\begin{align*}
  \minrep[\ctx]{\p} \succeq \minrep[\ctx']{\p} ~~&\textbf{if}~~ \ctx' \subseteq \ctx \\
  \minrep[\ctx]{\p} \succeq \minrep[\ctx]{\q} ~~&\textbf{if}~~ \joinprime{\q}  \\
\end{align*}
This is a well-founded relation (no infinite descending chain):
because of the uniqueness of $\joinprime{\p}$,
fixing any specific $\ctx$, a chain cannot be longer than 2.
So each chain starting at $\minrep[\ctx]{\p}$ is no longer than
$2|\ctx|$, which is finite.
As \cref{rule:min-pick} and \cref{rule:min-factor} are both decreasing,
the algorithm eventually reaches the base case and then terminates.

Minimal elements are unique and \cref{fig:min-max-algorithm} always terminates,
so it suffices to prove soundness.
More concretely, whenever \cref{fig:min-max-algorithm} derives
\( \minrep[\ctx]{\p} = \p' \),
we need to show
\( \ctx \models \p \sactsfor \q \iff \p' \actsfor \q \)
for all \q.

Let \( \q \in \p! \). We proceed by induction on the derivation.
\begin{pcases}
  \prule{rule:min-base}
    We have
    \( \ctx = (\p_1 \delegatedBy \q_1, \dots, \p_n \delegatedBy \q_n) \),
    \( \minrep[\ctx]{\p} = \p \),
    \( \joinprime{\p} \),
    and
    \( \forall i \in \enumSet{n} \centerdot \p \nactsfor \p_i \).

    \begin{pcases}
      \pcase{$\implies$}
        Assume
        \( \ctx \models \p \sactsfor \q \).
We apply \cref{thm:delegation-to-join} \( n \) times and
        \cref{thm:no-delegation} once to get
        \( \p \actsfor \q \vee \p_1 \vee \dots \vee \p_n \).
Since \( \joinprime{\p} \) and \( \p \nactsfor \p_i \) for all \( i \),
        it must be the case that \( \p \actsfor \q \) as desired.

      \pcase{$\impliedby$}
        Immediate by \cref{thm:syntactic-to-semantic-acts-for}.
    \end{pcases}

  \prule{rule:min-pick}
    We have
    \( \ctx = (\ctx', \p' \delegatedBy \q') \),
    \( \minrep[\ctx]{\p} = \minrep[\ctx']{\p \wedge \q'} \),
    \( \joinprime{\p} \),
    and
    \( \p \actsfor \p' \).

    \begin{pcases}
      \pcase{$\implies$}
        Assume
        \( \ctx \models \p \sactsfor \q \).
We need to show
        \(  \minrep[\ctx']{\p \wedge \q'} \actsfor \q \).
\Cref{thm:delegation-to-meet} gives
        \( \ctx' \models \p \wedge \q' \sactsfor \q \).
We can then apply the induction hypothesis to get the desired result.

      \pcase{$\impliedby$}
        Assume
        \(  \minrep[\ctx']{\p \wedge \q'} \actsfor \q \).
We need to show
        \( \ctx \models \p \sactsfor \q \).
The induction hypothesis gives
        \( \ctx' \models \p \wedge \q' \sactsfor \q \).
\Cref{thm:syntactic-to-semantic-acts-for} gives
        \( \ctx' \models \p \sactsfor \p' \).
        Combining these with
        \cref{thm:meet-to-delegation} gives the desired result.
    \end{pcases}

  \prule{rule:min-factor}
    We have
    \( \p = \p_1 \vee \dots \vee \p_n \),
    \( \minrep[\ctx]{\p} = \bigvee_{i \in \enumSet{n}}{\minrep[\ctx]{\p_i}} \),
    and
    \( \forall i \in \enumSet{n} \centerdot \joinprime{\p_i} \).

    \begin{pcases}
      \pcase{$\implies$}
        Assume
        \( \ctx \models \p \sactsfor \q \).
We need to show
        \( \bigvee_{i \in \enumSet{n}}{\minrep[\ctx]{\p_i}} \actsfor \q \).
By the definition of $\vee$, we have
        \( \ctx \models \p_i \sactsfor \q \)
        for all \( i \in \enumSet{n} \).
The induction hypotheses then give
        \( \minrep[\ctx]{\p_i} \actsfor \q \)
        for all \( i \in \enumSet{n} \).
Finally, we use the definition of $\vee$ to get the desired result.

      \pcase{$\impliedby$}
        Assume
        \( \bigvee_{i \in \enumSet{n}}{\minrep[\ctx]{\p_i}} \actsfor \q \).
We need to show
        \( \ctx \models \p \sactsfor \q \).
By the definition of $\vee$, we have
        \( \minrep[\ctx]{\p_i} \actsfor \q \)
        for all \( i \in \enumSet{n} \).
The induction hypotheses then give
        \( \ctx \models \p_i \sactsfor \q \)
        for all \( i \in \enumSet{n} \).
Finally, we use the definition of $\vee$ to get
        \( \ctx \models \p_1 \vee \dots \vee \p_n \sactsfor \q \).
    \end{pcases}
\end{pcases}
 \end{proof}

Our NMIF algorithm simply combines
\cref{fig:actsfor-algorithm,fig:min-max-algorithm}.

\begin{algorithm}[Uncompromised Label Check
  \( \Ctx \vdash \uncompromised{\lbl} \)]
  
  \label{fig:nmifc-algorithm}
  \begin{mathpar}
  \inferrule
    {
      \cctx \vdash \minrep[\ictx]{\ic} \sactsfor \cc
    }
    {
      \pair{\cctx}{\ictx} \vdash \uncompromised{\pair{\cc}{\ic}}
    }
\end{mathpar}
 \end{algorithm}

\begin{theoremrep}[Correctness of Uncompromised Label Check]
  \Cref{fig:nmifc-algorithm} terminates,
  and it is sound and complete with respect to
  \cref{def:uncompromised-labels}.
\end{theoremrep}
\begin{proof}
  We show soundness and completeness separately.
\begin{pcases}
  \pcase{Soundness}
    Assume
    \( \cctx \vdash \minrep[\ictx]{\ic} \sactsfor \cc \).
Pick
    \( \rr = \minrep[\ictx]{\ic} \).
By \cref{def:min-element}
    (instantiated with
    \( \minrep[\ictx]{\ic} \actsfor \minrep[\ictx]{\ic} \)),
    \( \ictx \models \ic \sactsfor \minrep[\ictx]{\ic} \).
    By \cref{thm:acts-for-correctness},
    \( \cctx \models \minrep[\ictx]{\ic} \sactsfor \cc \).

  \pcase{Completeness}
    Assume there exists \( \rr \in \p! \) such that
    \( \ictx \models \ic \sactsfor \rr \)
    and
    \( \cctx \models \rr \sactsfor \cc \).
By \cref{def:min-element},
    \( \minrep[\ictx]{\ic} \actsfor \rr \).
By \cref{thm:attacker},
    \( \cctx \models \minrep[\ictx]{\ic} \sactsfor \rr \).
By transitivity,
    \( \cctx \models \minrep[\ictx]{\ic} \sactsfor \cc \).
Finally, by \cref{thm:acts-for-correctness},
    \( \cctx \vdash \minrep[\ictx]{\ic} \sactsfor \cc \).
\end{pcases}
 \end{proof}

\section{Label Inference}
\label{sec:inference}

In practical IFC systems, label inference is an important way to avoid
redundant user annotations, since it is secure to infer labels
of intermediate computations from their inputs and outputs.
Typically, label inference is performed by solving a system of constraints
over lattice elements~\cite{jif,sherrloc}.
IFC systems further benefit from \emph{bounded label polymorphism},
which allows user to write library code reusable at different security levels.
In this section, we show how to do label inference
directly over the algebraic model of
labels using the algorithms from \cref{sec:algorithm}.

\begin{figure}
  \begin{minipage}[t]{0.48\textwidth}
    \noindent
\begin{lstlisting}[style=custom-small,lastline=8]
host Alice, Bob, Chuck
assume Alice = Bob for integrity
assume Bob = Chuck for integrity

fun average(a: int, b: int): int
{
  return (a + b) / 2
}
\end{lstlisting}
  \end{minipage}
  \hskip 2.5em
  \begin{minipage}[t]{0.43\textwidth}
    \noindent
\begin{lstlisting}[style=custom-small,firstline=10,firstnumber=8]
fun main() {
  val a = Alice.input
  val b = Bob.input
  val c = Chuck.input

  val r1 = average(a, b)
  val r2 = average(b, c)
}
\end{lstlisting}
  \end{minipage}
  \caption{The \lstinline{average} function is implicitly polymorphic
    over the labels of its arguments.}
  \label{fig:polymorphism-average}
\end{figure}

\subsection{Bounded Label Polymorphism}
\label{sec:label-polymorphism}

As in traditional type systems, allowing code that is generic over labels
increases expressiveness significantly.
Existing IFC-based languages like Jif~\cite{jif} and Flow Caml~\cite{flowcaml}
support bounded label polymorphism, which allows functions to be
parameterized over labels that are bounded by specified security levels.
Annotation burden on users can be further reduced by assuming information
flow from function arguments to return values by default.

In \cref{fig:polymorphism-average},
the annotation-free polymorphic function \lstinline{average} is
applied in \lstinline{main} to arguments with different
security labels.
Shown below is the same function with explicit annotations,
all of which can be inferred.
\begin{lstlisting}[style=custom-small]
fun average[X, Y, Z](a: int{X}, b: int{Y}): int{Z}
    where (X $\join$ Y $\flowsto$ Z)
\end{lstlisting}

Label inference assigns existential \emph{label variables}
to all unlabeled expressions and creates constraints
based on IFC typing rules.
The constraints are then solved by a constraint solver that
uses parameter bounds as delegation contexts.

In each function, polymorphic label variables are treated as label constants,
so solutions of label variables are expressed by both label constants
and polymorphic label variables.
Type checking at call sites ensures that
the parameter bounds are satisfied.
As functions have their own delegation contexts,
a new constraint system is solved for each function.

\subsection{Constraint Solver}
We describe a novel constraint solver that computes the
minimum-semantic-authority solution to constraint
systems with delegation contexts.
Minimum-authority solutions are desirable because they allow systems
to choose cheaper security enforcement mechanisms
\cite{zznm02,zcmz03,swift07,fr-popl08,fgr09,viaduct-pldi21}.

\begin{figure}
  \begin{minipage}[t]{.49\textwidth}
    \noindent
    \begin{syntax}
  \abstractCategory[L. Constants]{\lbl}
  \abstractCategory[L. Variables]{\lvar}

  \separate

  \category[L. Expressions]{\lexp}
  \alternative{\lbl}
  \alternative{\lvar}
  \alternative{\project{\lexp}} \\
  \alternative{\lexp_1 \join \lexp_2}
  \alternative{\lexp_1 \meet \lexp_2} \\
  \alternative{\lexp_1 \vee \lexp_2}
  \alternative{\lexp_1 \wedge \lexp_2}

  \category[L. Constraints]{\lcon}
  \alternative{\lexp_1 \flowsto \lexp_2}
  \alternative{\uncompromised{\lexp}}

  \categoryFromSet[Projections]{\proj}{\set{\cproj, \iproj}}
\end{syntax}
     \caption{Syntax of label constraints.}
    \label{fig:label-constraints}
  \end{minipage}
  \hfill
  \begin{minipage}[t]{.49\textwidth}
    \noindent
    \begin{syntax}
  \abstractCategory[P. Constants]{\p}
  \abstractCategory[P. Variables]{\pvar}

  \separate

  \category[P. Expressions]{\pexp}
  \alternative{\p}
  \alternative{\pvar}\\
  \alternative{\pexp_1 \vee \pexp_2}
  \alternative{\pexp_1 \wedge \pexp_2}\\
  \alternative{\p_1 \trustimply \pexp_2}
  \alternative{\minrep[\proj']{\pexp[\proj']}}

  \category[P. Constraints]{\pcon}
  \alternative{\pexp_1 \pactsfor \pexp_2}
\end{syntax}
     \vskip 1em
    \caption{Syntax of principal constraints.}
    \label{fig:principal-constraints}
  \end{minipage}
\end{figure}

\subsubsection{Label and Principal Constraints}
\Cref{fig:label-constraints} gives the syntax of the label constraint language.
Expressions in the constraint language include label constants and label variables,
authority projections,
as well as standard lattice operations.
A constraint either asserts that an expression flows to another,
or asserts that an expression is uncompromised.
We translate constraints over labels to constraints over principals
to leverage algorithms from \cref{sec:algorithm}.
\Cref{fig:principal-constraints} gives the syntax of the principal
constraint language.
The syntax includes a principal variable \pvar[\proj]
for each combination of label variable \lvar and projection \proj.
That is, \pvar[\cproj] represents the confidentiality of \lvar,
and \pvar[\iproj] represents \lvar's integrity.

We index expressions \pexp by the component \proj they represent.
This prevents expressions like $\pvar[\cproj]_1 \wedge \pvar[\iproj]_2$,
whose components are mixed.
Expressions include principal constants and principal variables,
as well as principal-lattice operations $\vee$ and $\wedge$.
The operation $\trustimply$ is called the
\emph{relative pseudocomplement} of the meet operation:
$\p_1 \trustimply \p_2$ is defined as the minimum-authority principal \p such that
$\p_1 \wedge \p \actsfor \p_2$.
We use $\trustimply$ to solve constraints of the form
$\pvar \wedge \p_1 \pactsfor \pexp_2$.The $\minrep[\proj]{}$ operation allows mixing integrity and confidentiality components;
we use it when solving for labels that must be uncompromised.
Principal constraints have the form
$\pexp_1 \pactsfor \pexp_2$,
which stands for $\ctx_{\proj} \models \pexp_1 \sactsfor \pexp_2$.

\Cref{fig:label-to-principal-constraints} gives rules for translating
label constraints to principal constraints.
The definition of $\getcomponent{\lexp}$ is a straightforward encoding of the
label-lattice operations.
Using $\pconstraint{\lexp}$,
we translate a flows-to ($\flowsto$) constraint to two
acts-for ($\actsfor$) constraints, one for each label component.
The constraint $\uncompromised{\lexp}$ follows from \cref{fig:nmifc-algorithm}.

\begin{figure}
  \begingroup \newcommand{\also}{,\;}\begin{judgments}
  \judgment{\pconstraint{\lcon} = \pcon_1, \dots, \pcon_n}
\end{judgments}
\begin{mathpar}
  \pconstraint{\lexp_1 \flowsto \lexp_2}
  =
  \getcomponent[\cproj]{\lexp_2}
    \Cactsfor
    \getcomponent[\cproj]{\lexp_1}
  \also
  \getcomponent[\iproj]{\lexp_1}
    \Iactsfor
    \getcomponent[\iproj]{\lexp_2}

  \pconstraint{\uncompromised{\lexp}}
  =
  \getcomponent[\iproj]{\lexp}
    \Iactsfor
    \minrep[\cproj]{\getcomponent[\cproj]{\lexp}}
\end{mathpar}
\endgroup
   \newcommand{\getcomponentprotected}[1]{{\getcomponent{#1}}}\begin{judgments}
  \judgment{\getcomponent{\lexp} = \pexp}
\end{judgments}
\hfill
\begin{minipage}{.48\textwidth}
  \begin{function}{\getcomponentprotected}
    \case{\lexp_1 \join \lexp_2}
    \getcomponent{
      \confidentiality{(\lexp_1 \wedge \lexp_2)}
      \wedge
      \integrity{(\lexp_1 \vee \lexp_2)}
    }

    \case{\lexp_1 \meet \lexp_2}
    \getcomponent{
      \confidentiality{(\lexp_1 \vee \lexp_2)}
      \wedge
      \integrity{(\lexp_1 \wedge \lexp_2)}
    }

    \case{\lexp_1 \vee \lexp_2}
    \getcomponent{\lexp_1} \vee \getcomponent{\lexp_2}

    \case{\lexp_1 \wedge \lexp_2}
    \getcomponent{\lexp_1} \wedge \getcomponent{\lexp_2}
  \end{function}
\end{minipage}
\hfill
\begin{minipage}{.48\textwidth}
  \begin{function}{\getcomponentprotected}
    \case{\pair{\p}{\q}}
    \begin{cases}
      \p & \text{if } \proj = \cproj
      \\
      \q & \text{if } \proj = \iproj
    \end{cases}

    \case{\lvar}
    \pvar[\proj]

    \case{\project[\proj']{\lexp}}
    \begin{cases}
      \getcomponent{\lexp} & \text{if } \proj = \proj'
      \\
      \weakest & \text{if } \proj \neq \proj'
    \end{cases}
  \end{function}
\end{minipage}
   \caption{Translating label constraints to principal constraints.}
  \label{fig:label-to-principal-constraints}
\end{figure}

\subsubsection{Solving Principal Constraints}

Our constraint solver requires the left-hand side of each constraint to be
atomic (a constant or a variable), that is, constraints
of the form
\( \p_1 \pactsfor \pexp_2 \)
and
\( \pvar_1 \pactsfor \pexp_2 \).
We exhaustively apply the equational axioms of Heyting algebra
(e.g., associativity, absorption, distributivity, etc.)
until no left-hand side of any constraint can be further simplified.
As equational axioms are syntactic rewrites,
it does not change the constraint system over the underlying algebra.
Moreover, this process always terminates,
and ensures that the left-hand side of each constraint either
is atomic or contains a meet ($\wedge$).\footnote{
  Translation rules in \cref{fig:label-to-principal-constraints}
  never generate constraints with $\trustimply$ or $\minrep[\proj]{\cdot}$
  on the left-hand side,
  and the constraint simplification eliminate all joins ($\vee$) on the left.
}

Constraint solving fails if the left-hand side of any constraint contains
a meet, since such systems do not have unique solutions.
For example, the system
\( \lvar_1 \wedge \lvar_2 \actsfor \alice \)
has no minimal solution:
we can assign
\( \set{ \lvar_1 \mapsto \alice, \lvar_2 \mapsto \weakest } \)
or
\( \set{ \lvar_1 \mapsto \weakest, \lvar_2 \mapsto \alice } \),
but neither solution is better than the other.
Compiler implementations need to either restrict the syntax of
polymorphic constraints, or report errors during constraint solving.

Once the simplification process succeeds, we extend
the algorithm of Rehof and Mogensen~\cite{RM96} for iteratively solving semi-lattice
constraints.
We initialize all principal variables to $\weakest$, and use unsatisfied
constraints to update variables repeatedly until a fixed point is reached,
using the rule:
\[
  \text{given}
  \quad
  \pvar \pactsfor \pexp
  ,
  \quad
  \text{set}
  \quad
  \pvar \coloneq \pvar \wedge \currentval{\pexp}
  ,
\]
where $\currentval{\pexp}$ is the value of \pexp according to the current
assignment.

Note that constraints that have constants \p on the left-hand
side are ignored during the fixed point computation.
Once a fixed-point solution is reached,
we perform the following check for each constraint with a constant left-hand side:
\[
  \text{given}
  \quad
  \p \pactsfor \pexp
  ,
  \quad
  \text{check}
  \quad
  \ctxof \models \p \sactsfor \currentval{\pexp}
  .
\]
We state and prove in the technical report \cite{ifc-delegation-tr}
that this process terminates with the minimum-authority solution
if all such constraints are satisfied;
otherwise, there is no valid solution.
\begin{theoremrep}[Correctness of Constraint Solver]
  \label{thm:constraint-solver-correctness}
  Constraint procedure always terminates with either minimal-authority solution
  or correctly fails.
\end{theoremrep}
\begin{proof}
  We note that the result of the algorithm is a map $M$ from
principal variables $\pvar$ to principal constants $\p$.

Now define a well-founded relation $\succeq$ over
the domain of maps $M$ from $\pvar$ to $\p$
where the update rule is decreasing:
\[M \succeq M' ~~\mathsf{if}~~ \forall \pvar
\centerdot M(\pvar) \actsfor M'(\pvar)\]
This is indeed a well-founded relation in practice
when $M$ is a finite map and $\p!$ is a finite lattice.

Then we formalize the update rule as a
function $f$ between maps $M$ given by
(here $D$ is the set of constraints):
\[f(M)(\pvar) = \bigwedge_{\pvar \pactsfor \pexp \in D}
{M(\pvar) \wedge M(\pexp)}\]

$f$ is a monotonic function by construction.
Since our algorithm repeatedly applies this function to $M$,
this forms a chain of $f^i(M)$, which, by Knaster--Tarski theorem,
reaches the greatest fixpoint.
Thus, this algorithm terminates with least-authority solution when it exists.
 \end{proof}

\subsection{Implementation}
\label{sec:implementation}

We modified the parser of the Viaduct compiler~\cite{viaduct-pldi21}
with the delegation syntax, and extended its
static analysis procedure with our label inference algorithm.
\footnote{
  Code available at:
  \ifblinded
  -anonymous link-
  \else
  \href{https://github.com/apl-cornell/viaduct}{https://github.com/apl-cornell/viaduct}.
  \fi
}

Because the original Viaduct constraint solver can only make syntactic comparison
($\actsfor$) between labels, it instantiates polymorphic variables
with constant principal names.
Therefore, Viaduct has to run a \emph{specialization procedure} to create
a monomorphic copy of a function at each call site.
In nested function calls, the number of monomorphic functions
created grows exponentially with the depth of calls.
Recursive calls need to be handled explicitly to ensure termination.

Thanks to delegation contexts, label inference no longer
requires monomorphic functions and can be done in one pass.
For each function call site, the specialization procedure
memoizes the argument labels,
and avoids creating duplicates of monomorphic functions that are instantiated
with the same polymorphic argument labels,
eliminating the need to treat recursive functions separately.
Without duplication, our
procedure creates a minimum-sized set monomorphic functions.

Empirically, type inference is fast. On all of the Viaduct
benchmarks, it terminates within 300
milliseconds on a
Macbook air 2023 with an M2 chip.

\section{Related Work}
\label{sec:related-works}

Expressive trust delegation has been widely investigated in
access control and capability systems~\cite{rbac-iflow,
li2003delegation},
where delegation is called \emph{role hierarchy}
\cite{rbac-iflow}.
Such systems support role adoption, which is a form of dynamic delegation.
However, access control systems generally do not connect to
information flow security properties.
Many prior IFC systems have delegation abstractions compatible with our framework,
but they either lack support for asymmetric delegation~\cite{nmifc,
viaduct-pldi21, zsm19}, or do not explore its effect
over hyperproperties~\cite{zm01b, msz04, msz06, flam, li2021generalpurposedynamicinformationflow}.

Our inference algorithm differs from prior implementations of syntax-directed IFC
label inference~\cite{jif, sherrloc, viaduct-pldi21} by operating
directly over the underlying algebra.
This algebraic approach enhances extensibility:
adding principals or delegations does not invalidate existing analysis.
\citet{Dolan}{algebraicSubtyping} proposes an inference algorithm for an
expressive algebraic type system with
type constructors and recursive function types.
However, their type system is inherently complex,
and the inference doesn’t produce easily interpretable representions
of types.
In contrast, our label system balances expressiveness with a
simple and effective inference algorithm.

FLAM~\cite{flam} proposes a label model and defines
\emph{robust authorization} to address security vulnerabilities
arising from dynamic delegation.
However, robust authorization is a proof-theoretic definition with no
semantic model.
\Citet{Arden and Myers}{flac} propose the FLAC calculus, based on the
FLAM label model, and prove robust declassification for a language that downgrades using
dynamic delegation.
Similarly, FLAC has no attacker semantics.
\Citet{Cecchetti et al.}{nmifc} define NMIF, but their
system cannot explicitly express delegations between atomic principals.

\section{Conclusion and Future Directions}
\label{sec:conclusion}

We present an algebraic semantic framework for IFC labels that models asymmetric delegation,
along with sound and complete algorithms
that enforce security properties and infer security annotations.
Our approach provides a solid foundation
for building modular, expressive and extensible IFC systems.

\ifacknowledgments
  \begin{credits}
    \subsubsection{\ackname}
This work was supported by the
National Science Foundation under
NSF grant 1704788. Any opinions, findings, conclusions, or recommendations expressed in this
material are those of the authors and do not necessarily reflect the views
of the National Science Foundation.
We also thank
Suraaj Kanniwadi
and
Ethan Cecchetti
for discussions and feedback.

\subsubsection{\discintname}

The authors have no competing interests to declare that are relevant to the
content of this article.
   \end{credits}
\fi

\finalpage

\bibliography{bibtex/pm-master}

\end{document}